\documentclass{llncs}
\usepackage{epsfig}
\usepackage{amssymb}
\usepackage{algorithmic}

\def\idtt#1{\ensuremath{\mathtt{#1}}}

\def\rankop{\idtt{rank}}
\def\selop{\idtt{select}}
\def\accessop{\idtt{access}}
\def\insertop{\idtt{insert}}
\def\deleteop{\idtt{delete}}
\def\sumop{\idtt{sum}}
\def\searchop{\idtt{search}}
\def\updateop{\idtt{update}}

\def\polylog{\idtt{polylog}}

\def\etal{{\em et al.}}

\begin{document}

\title{Succinct Representations of Dynamic Strings}
\author{Meng He \and J. Ian Munro}
\institute{
  Cheriton School of Computer Science,
  University of Waterloo, Canada, \\
  \email \{mhe, imunro\}@uwaterloo.ca
}

\maketitle

\setcounter{page}{1}
\pagestyle{plain}
\thispagestyle{plain}

\begin{abstract}
The $\rankop$ and $\selop$ operations over a string of length $n$ from an alphabet of size $\sigma$ have been used widely in the design of succinct data structures. 
In many applications, the string itself must be maintained dynamically, allowing characters of the string to be inserted and deleted. 
Under the word RAM model with word size $w=\Omega(\lg n)$, we design a succinct representation of dynamic strings using $nH_0 + o(n)\cdot\lg\sigma + O(w)$ bits to support $\rankop$, $\selop$, $\insertop$ and $\deleteop$ in $O(\frac{\lg n}{\lg\lg n}(\frac{\lg \sigma}{\lg\lg n}+1))$ time\footnote{$\lg n$ denotes $\log_2 n$.}. 
When the alphabet size is small, i.e. when $\sigma = O(\polylog (n))$, including the case in which the string is a bit vector, these operations are supported in $O(\frac{\lg n}{\lg\lg n})$ time. 
Our data structures are more efficient than previous results on the same problem, and we have applied them to improve results on the design and construction of space-efficient text indexes. 
\end{abstract}

\section{Introduction}
\label{sec:intro}

Succinct data structures provide solutions to reduce the storage cost of modern applications that process large data sets, such as web search engines, geographic information systems, and bioinformatics applications. 
First proposed by Jacobson~\cite{j1989}, the aim is to encode a data structure using space close to the information-theoretic lower bound, while supporting efficient navigation operations in them. 
This approach was successfully applied to many abstract data types, including bit vectors~\cite{j1989,cm1996,rrs2007}, strings~\cite{ggv2003,bhmr2007,gn2009}, binary relations~\cite{bgmr2007,bhmr2007}, (unlabeled and labeled) trees~\cite{j1989,mr2001,grr2006,flmm2009,bgmr2007,bhmr2007,sn2010}, graphs~\cite{j1989,mr2001,bchm2007} and text indexes~\cite{ggv2003,chl2004,gn2009}. 

A basic building block for most succinct data structures is the pair of operations $\rankop$ and $\selop$. 
In particular, we require a highly space-efficient representation of a string $S$ of length $n$ over an alphabet of size $\sigma$ to support the fast evaluation of: 
\begin{itemize}
\item $\accessop(S, i)$, which returns the character at position $i$ in the string $S$;
\item $\rankop_{\alpha}(S, i)$, which returns the number of occurrences of character $\alpha$ in $S[1..i]$; 
\item $\selop_{\alpha}(S, i)$, which returns the position of the $i$\textsuperscript{th} occurrence of character $\alpha$ in the string $S$. 
\end{itemize}
This problem has many applications such as designing space-efficient text indexes~\cite{ggv2003,gn2009}, as well as representing binary relations~\cite{bgmr2007,bhmr2007}, labeled trees~\cite{flmm2009,bgmr2007,bhmr2007} and labeled graphs~\cite{bchm2007}. 
The case in which the string is a bit vector whose characters are $0$'s and $1$'s (i.e. $\sigma = 2$) is even more fundamental:  
A bit vector supporting $\rankop$ and $\selop$ is a key structure used in several approaches of representing strings succinctly~\cite{ggv2003,bgmr2007,bhmr2007,mn2008}, and it is also used in perhaps most succinct data structures~\cite{j1989,mr2001,grr2006,flmm2009,bgmr2007,bhmr2007,bchm2007}. 

Due to the importance of strings and bit vectors, researchers have designed various succinct data structures for them~\cite{j1989,rrs2007,ggv2003,bhmr2007,gn2009} and achieved good results. 
For example, the data structure of Raman~\etal~\cite{rrs2007} can encode a bit vector using $nH_0 + o(n)$ bits, where $H_0$ is the zero-order entropy of the bit vector\footnote{The zero-order (empirical) entropy of a string of length $n$ over an alphabet of size $\sigma$ is defined as $H_0 = \sum_{i=1}^{\sigma}(\frac{n_i}{n}\lg \frac{n}{n_i})$, where $n_i$ is the number of times that the $i$\textsuperscript{th} character occurs in the string. Note that we always have $H_0 \le \lg\sigma$. This definition also applies to a bit vector, for which $\sigma = 2$.}, to support {\accessop}, {\rankop} and {\selop} operations in constant time. 
Another data structure called wavelet tree proposed by Grossi~\etal~\cite{ggv2003} can represent a string using $nH_0+o(n) \cdot \lg\sigma$ bits to support $\accessop$, $\rankop$ and $\selop$ in $O(\lg \sigma)$ time. 

However, in many applications, it is not enough to have succinct static data structures that allow data to be retrieved efficiently, because data in these applications are also updated frequently. 
For example, when designing a succinct text index to support fast string search in a collection of text documents, it might be necessary to allow documents to be inserted into or deleted from the collection. 
Thus it is necessary to study succinct dynamic data structures that not only support efficient retrieval operations, but also support efficient update operations. 
In the case of strings and bit vectors, the following two update operations are desired in many applications in addition to $\accessop$, $\rankop$ and $\selop$: 
\begin{itemize}
\item $\insertop_{\alpha}(S, i)$, which inserts character $\alpha$ between $S[i-1]$ and $S[i]$;
\item $\deleteop(S, i)$, which deletes $S[i]$ from $S$. 
\end{itemize}

In this paper, we design succinct representations of dynamic strings and bit vectors that are more efficient than previous results. We also present several applications to show how advancements on these fundamental problems yield improvements on other data structures. 

\subsection{Related Work}
\label{sec:related}

Blandford and Blelloch~\cite{bb2004} considered the problem of representing ordered lists succinctly, and their result can be used to represent a dynamic bit vector of length $n$ using $O(nH_0)$ bits to support the operations defined in Section~\ref{sec:intro} in $O(\lg n)$ time (note that $H_0 \le 1$ holds for a bit vector). 
A different approach proposed by Chan~\etal~\cite{chl2004} can encode dynamic bit vectors using $O(n)$ bits to provide the same support for operations. 
Later Chan~\etal~\cite{chls2007} improved this result by providing $O(\lg n/\lg\lg n)$-time support for all these operations while still using $O(n)$ bits of space. 
M{\"a}kinen and Navarro~\cite{mn2008} reduced the space cost to $nH_0 + o(n)$ bits, but their data structure requires $O(\lg n)$ time to support operations. 
Recently, Sadakane and Navarro~\cite{sn2010} designed a data structure for dynamic trees, and their main structure is essentially a bit vector that supports more types of operations. 
Their result can be used to represent a bit vector using $n+o(n)$ bits to support the operations we consider in $O(\lg n / \lg\lg n)$ time. 

For the more general problem of representing dynamic strings of length $n$ over alphabets of size $\sigma$, M{\"a}kinen and Navarro~\cite{mn2008} combined their results on bit vectors with the wavelet tree structure of Grossi~\etal~\cite{ggv2003} to design a data structure of $nH_0 + o(n)\cdot\lg \sigma$ bits that supports $\accessop$, $\rankop$ and $\selop$ in $O(\lg n \log_q\sigma)$ time, and $\insertop$ and $\deleteop$ in $O(q\lg n \log_q\sigma)$ time for any $q = o(\sqrt{\lg n})$. 
Lee and Park~\cite{lp2009} proposed another data structure of $n\lg\sigma + o(n)\cdot\lg \sigma$ to support {\accessop}, {\rankop} and {\selop} in $O(\lg n (\frac{\lg\sigma}{\lg\lg n} +1))$ worst-case time which is faster, but {\insertop} and {\deleteop} take $O(\lg n (\frac{\lg\sigma}{\lg\lg n} +1))$ amortized time. 
Finally, Gonz\'{a}lez and Navarro~\cite{gn2009} improved the above two results by designing a structure of $nH_0 + \linebreak o(n)\cdot\lg \sigma$ bits to support all the operations in $O(\lg n (\frac{\lg\sigma}{\lg\lg n} +1))$ worst-case time.  

Another interesting data structure is that of Gupta~\etal~\cite{ghsv2007}. 
For the same problems, they aimed at improving query time while sacrificing update time. 
Their bit vector structure occupies $nH_0 + o(n)$ bits and requires $O(\lg\lg n)$ time to support {\accessop}, {\rankop} and {\selop}. It takes $O(\lg^{\epsilon} n)$ amortized time to support {\insertop} and {\deleteop} for any constant $0<\epsilon < 1$. 
Their dynamic string structure uses $n\lg\sigma + \lg\sigma(o(n)+O(1))$ bits to provide the same support for operations (when $\sigma = O(\polylog(n))$, {\accessop}, {\rankop} and {\selop} take $O(1)$ time). 

\subsection{Our Results}

We adopt the word RAM model with word size $w = \Omega(\lg n)$. 
Our main result is a succinct data structure that encodes a string of length $n$ over an alphabet of size $\sigma$ in $n H_0 + o(n)\cdot\lg\sigma + O(w)$ bits to support {\accessop}, {\rankop}, {\selop}, {\insertop} and {\deleteop} in $O(\frac{\lg n}{\lg\lg n}(\frac{\lg \sigma}{\lg\lg n}+1))$ time. 
When $\sigma = O(\polylog(n))$, all these operations can be supported in $O(\frac{\lg n}{\lg\lg n})$ time. 
Note that the $O(w)$ term in the space cost exists in all previous results, and we omit them in Section~\ref{sec:related} and Table~\ref{table:compare} for simplicity of presentation (in fact many papers simply ignore them). 
Our structure can also encode a bit vector of length $n$ in $nH_0 + o(n) + O(w)$ bits to support the same operations in $O(\frac{\lg n}{\lg\lg n})$ time, matching the lower bound in \cite{fs1989}. 
Table~\ref{table:compare} in Appendix~\ref{app:compare} compares these results with previous results, from which we can see that our solutions are currently the best to the problem, 
for both the general case and the special case in which the alphabet size is $O(\polylog(n))$ or $2$ (i.e. the string is a bit vector). 
The only previous result that is not comparable is that of Gupta~\etal~\cite{ghsv2007}, since their solution is designed under the assumption that the string is queried frequently but updated infrequently.


We also apply the above results to design a succinct text index for a dynamic text collection to support text search, and the problem of reducing the required amount of working space when constructing a text index. 
Our dynamic string representation allows us to improve previous results on these problems~\cite{mn2008,gn2009}. 

\section{Preliminaries}
\paragraph{Searchable Partial Sums.} Raman~\etal~\cite{rrr2001} considered the problem of representing a dynamic sequence of integers to support {\sumop}, {\searchop} and {\updateop} operations. 
To achieve their main result, they designed a data structure for the following special case in which the length of the sequence is small: 

\begin{lemma}
\label{lem:smallpartial}

There is a data structure that can store a sequence, $Q$, of $O(\lg^{\epsilon} n)$ nonnegative integers of $O(\lg n)$ bits each\footnote{Raman~\etal~\cite{rrr2001} required each integer to fit in one word. However, it is easy to verify that their proof is still correct if each integer requires $O(\lg n)$ bits, i.e. each integer can require up to a constant number of words to store.}, for any constant $0 \le \epsilon < 1$, using $O(\lg^{1+\epsilon}n)$ bits to support the following operations in $O(1)$ time: 
\begin{itemize}
\item $\sumop(Q, i)$, which computes $\sum_{j=1}^i Q[j]$;
\item $\searchop(Q, x)$, which returns the smallest $i$ such that $\sumop(Q, i) \ge x$;
\item $\updateop(Q, i, \delta)$, which updates $Q[i]$ to $Q[i] + \delta$, where $|\delta| \le \lg n$.
\end{itemize}
The data structure requires a precomputed universal table of size $O(n^{\epsilon'})$ bits for any fixed $\epsilon' > 0$. 
The structure can be constructed in $O(\lg^{\epsilon} n)$ time except the precomputed table. 
\end{lemma}

We will use this lemma to encode information stored as small sequences of integers in our data structures. 

\label{sec:cspsi}

\paragraph{Collections of Searchable Partial Sums.} A key structure in the dynamic string representation of Gonz\'{a}lez and Navarro~\cite{gn2009} is a data structure that maintains a set of sequences of nonnegative integers, such that $\sumop$, $\searchop$ and $\updateop$ can be supported on any sequence efficiently, while {\insertop} and {\deleteop} operations are performed simultaneously on all the sequences at the same given position, with the restriction that only $0$'s can be inserted or deleted. 
More precisely, let $C$ = $Q_1, Q_2, \cdots, Q_d$ to be a set of dynamic sequences, and each sequence, $Q_j$, has $n$ nonnegative integers of $k = O(\lg n)$ bits each. 
The {\em collection of searchable partial sums with insertions and deletions (CSPSI)} problem is to encode $C$ to support the following operations:
\begin{itemize}
\item $\sumop(C, j, i)$, which computes $\sum_{p=1}^i Q_j[p]$;
\item $\searchop(C, j, x)$, which returns the smallest $i$ such that $\sumop(C, j, i) \ge x$;
\item $\updateop(C, j, i, \delta)$, which updates $Q_j[i]$ to $Q_j[i] + \delta$;
\item $\insertop(C, i)$, which inserts $0$ between $Q_j[i-1]$ and $Q_j[i]$ for all $1 \le j \le d$;
\item $\deleteop(C, i)$, which deletes $Q_j[i]$ from sequence $Q_j$ for all $1 \le j \le d$, and to perform this operation, $Q_j[i] = 0$ must hold for all $1 \le j \le d$. 
\end{itemize}
Gonz\'{a}lez and Navarro~\cite{gn2009} designed a data structure of $kdn (1 + O(\frac{1}{\sqrt{\lg n}} + \frac{d}{\lg n}))$ bits to support all the above operations in $O(d+\lg n)$ time. 
This structure becomes succinct (i.e. using $dk(n+o(n))$ bits if $d = o(\lg n)$, when the operations can be supported in $O(\lg n)$ time. 
A careful reading of their technique shows that  these results only work under the word RAM model with word size $w=\Theta(\lg n)$ (see our discussions after Lemma~\ref{lem:cspsi}). 
We improve this data structure for small $d$, which is further used to design our succinct string representations.

\section{Collections of Searchable Partial Sums}
\label{sec:ourcspsi}

We follow the main steps of the approach of Gonz\'{a}lez and Navarro~\cite{gn2009} to design a succinct representation of dynamic strings, but we make improvements upon the data structures designed in each step. 
We first improve their result to design a data structure for the collection of searchable partial sums with insertions and deletions problem (Section~\ref{sec:ourcspsi}), and then combine it with other techniques to improve their data structure for strings over small alphabets (Section~\ref{sec:smallalphabet}). 
Finally, we extend the result on small alphabets to general alphabets using the structure of multi-ary wavelet trees (Section~\ref{sec:generalalphabet}). 
Our main strategy of achieving these improvements is to divide the sequences into superblocks of appropriate size, and store them in the leaves of a B-tree (instead of the red-black tree in~\cite{gn2009}). 
Similar ideas were applied to data structures for balanced parentheses~\cite{chls2007,sn2010}. Our work is the first that successfully adapts it to integer sequences and character strings, and we have created new techniques to overcome some difficulties. 

In this section, we consider the CSPSI problem defined in section~\ref{sec:cspsi}. We assume that $d = O(\lg^{\eta} n)$ for any constant $0<\eta<1$, and for the operation $\updateop(C, j, i, \delta)$, we assume $|\delta| \le \lg n$. Under these assumptions, we improve the result in \cite{gn2009} 
under the word RAM model with word size $\Omega(\lg n)$. 

\label{sec:cspsids}

\paragraph{Data Structures.} Our main data structure is a B-tree constructed over the given collection $C$. 
Let $L = \lceil\frac{\lceil\lg n\rceil^2}{\lg\lceil\lg n\rceil}\rceil$. Each leaf of this B-tree stores a {\em superblock} whose size is between (and including) $L/2$ and $2L$ bits, and each superblock stores the same number of integers from each sequence in $C$. 
More precisely, the content of the leftmost leaf is $Q_1[1..s_1]Q_2[1..s_1]\cdots Q_d[1..s_1]$, the content of the second leftmost leaf is $Q_1[s_1+1..s_2]Q_2[s_1+1..s_2]\cdots Q_d[s_1+1..s_2]$, and so on, and the indices $s_1, s_2, \cdots$ satisfy the following conditions because of requirement on the sizes of superblocks: $L/2 \le s_1 k d \le 2L, L/2 \le (s_2-s_1) k d \le 2L,\cdots$. 
 
Let $f = \lg^{\lambda} n$, where $\lambda$ is a positive constant number less than $1$ that we will fix later. Each internal node of the B-tree we construct has at least $f$ and at most $2f$ children. 
We store the following $d+1$ sequences of integers for each internal node $v$ (let $h$ be the number of children of $v$): 
\begin{itemize}
\item A sequence $P(v)[1..h]$, in which $P(v)[i]$ is the number of positions stored in the leaves of the subtree rooted at the $i$\textsuperscript{th} child of $v$ for any sequence in $C$ (note that this number is the same for all sequences in $C$);
\item A sequence $R_j(v)[1..h]$ for each $j = 1, 2, \cdots, d$, in which $R_j(v)[i]$ is the sum of the integers from sequence $Q_j$ that are stored in the leaves of the subtree rooted at the $i$\textsuperscript{th} child of $v$. 
\end{itemize}
We use Lemma~\ref{lem:smallpartial} to encode each of the $d+1$ sequences of integers for $v$. 

We further divide each superblock into {\em blocks} of $\lceil\lceil \lg n\rceil^{3/2}\rceil$ bits each, and maintain the blocks for the same superblock using a linked list. 
Only the last block in the list can be partially full; any other block uses all its bits to store the data encoded in the superblock. 
This is how we store the superblocks physically. 


We provide a proof sketch for the following lemma that analyzes the space cost of the above data structures (See Appendix~\ref{app:cspsisize} for the full proof): 
\begin{lemma}
\label{lem:cspsisize}
The above data structures occupy $kd (n + o(n))$ bits if the parameters $\lambda$ and $\eta$ satisfy $0<\lambda < 1-\eta$. 
\end{lemma}

\begin{proof}[sketch]
It takes $kd n$ bits to encode the integers in $C$. 
The number of pointers that link the blocks is linear in the number of blocks which is  $O(\frac{kd n}{\lg^{3/2} n})$, so all the pointers occupy $O(\frac{kd n}{\sqrt{\lg n}})$ bits in total. 
The space wasted in the blocks that are partially full is $O(\frac{kd n \lg\lg n}{\sqrt{\lg n}})$ bits. 
The number of nodes in the B-tree is linear in the number of superblocks which is $O(\frac{kd n \lg\lg n}{\lg^2 n})$. 
The space used for each internal node is dominated by the space cost of the $d+1$ sequences constructed for it, which is $O(df\lg n)$ bits with a precomputed universal table of size $o(n)$ bits (we only need one copy of this {\em universal} table for the sequences for all the internal nodes). 
Therefore, the total space cost in bits is $kd n (1+ O(\frac{\lg\lg n}{\sqrt{\lg n}}) + O(\frac{df \lg\lg n}{\lg n})) = kd n (1+ O(\frac{\lg\lg n}{\sqrt{\lg n}}) +  O(\frac{\lg\lg n}{\lg^ {1-\eta - \lambda}n}))$,  
which is $kd (n+o(n))$ when $0<\lambda < 1-\eta$. \qed
\end{proof}


\paragraph{Supporting {\sumop}, {\searchop} and {\updateop}.} We discuss these three operations first because they do not change the size of $C$. 

\begin{lemma}
\label{lem:cspsisum}
The data structures in this section 
can support {\sumop}, {\searchop} and {\updateop} in $O(\frac{\lg n}{\lg\lg n})$ time with an additional universal table of $o(n)$ bits. 
\end{lemma}

\begin{proof}
To support $\sumop(C, j, i)$, we perform a top-down traversal in the B-tree. 
In our algorithm, we use a variable $r$ that is initially $0$, and its value will increase as we go down the tree. 
We have another variable $s$ whose initial value is $i$. 
Initially, let $v$ be the root of the tree. 
As $P(v)$ stores the number of positions stored in the subtrees rooted at each child of $v$, the subtree rooted at the $c$\textsuperscript{th} child of $v$, where $c=\searchop(P(v), i)$, contains position $i$. 
We also compute the sum of the integers from the sequence $Q_j$ that are stored in the subtrees rooted at the left siblings of the $c$\textsuperscript{th} child of $v$, which is $y=\sumop(R_j(v), c-1)$, and we increase the value of $r$ by $y$. 
We then set $v$ to be its $c$\textsuperscript{th} child, decrease the value of $s$ by $\sumop(P(v), c-1)$, and the process continues until we reach a leaf. 
At this time, $r$ records the sum of the integers from the sequence $Q_j$ that are before the first position stored in the superblock of the leaf we reach. 
As the height of this B-tree is $O(\frac{\lg n}{\lg\lg n})$, and the computation at each internal node takes constant time by Lemma~\ref{lem:smallpartial}, it takes $O(\frac{\lg n}{\lg\lg n})$ time to locate this superblock. 

It now suffices to compute the sum of the first $s$ integers from sequence $Q_j$ that are stored in the superblock. 
This can be done by first going to the block storing the first integer in the superblock that is from $Q_j$, which takes $O(\frac{\sqrt{\lg n}}{\lg\lg n})$ time (recall that each block is of fixed size and there are $O(\frac{\sqrt{\lg n}}{\lg\lg n})$ of them in a superblock), and then read the sequence in chunks of $\lceil \frac{1}{2}\lg n \rceil$ bits. 
For each $\lceil \frac{1}{2}\lg n \rceil$ bits we read, we use a universal table $A_1$ to find out the sum of the $z=\lfloor\lceil \frac{1}{2}\lg n \rceil / k\rfloor$ integers stored in it (the last $a=\lceil \frac{1}{2}\lg n \rceil \bmod k$ bits in this block are concatenated with the next $\lceil \frac{1}{2}\lg n \rceil - a$ bits read for table lookup). 
This table simply stores the result for each possible bit strings of length $\lceil \frac{1}{2}\lg n \rceil$, so that the above computation can be done in constant time. 
Note that the last chunk we read this way may contain integers after $Q_j[i]$. 
To address the problem, we augment $A_1$ so that it is a two dimensional table $A_1[1..2^{\lceil \frac{1}{2}\lg n \rceil}][1..z]$, in which $A[b][g]$ stores for the $b$\textsuperscript{th} lexicographically smallest bit vector of length $\lceil \frac{1}{2}\lg n \rceil$, the sum of the first $g$ integers of size $k$ stored in it. 
This way the computation in the superblock can be done in $O(\frac{\lg n}{\lg\lg n})$ time, and thus  $\sumop(C, j, i)$ can be supported in  $O(\frac{\lg n}{\lg\lg n})$ time. 
The additional data structure we require is table $A_1$, which occupies $O(2^{\lceil \frac{1}{2}\lg n \rceil} \times \lfloor\lceil \frac{1}{2}\lg n \rceil / k\rfloor \times \lg n) = O((\sqrt{n}\lg^2 n)/k)$ bits. 

The operations $\searchop$ and $\updateop$ can be supported in a similar manner. See Appendix~\ref{app:cspsisearch} for more details. \qed
\end{proof}

\paragraph{Supporting {\insertop} and {\deleteop}.} We give a proof sketch of the following lemma on supporting {\insertop} and {\deleteop} (See Appendix~\ref{app:cspsiinsert} for the full proof): 
\begin{lemma}
\label{lem:cspsiinsert}
When $w = \Theta(\lg n)$, the data structures in this section can support {\insertop} and {\deleteop} in $O(\frac{\lg n}{\lg\lg n})$ amortized time. 
\end{lemma}
\begin{proof}[sketch]
To support $\insertop(C, i)$, we locate the leaf containing position $i$ as we do for $\sumop$, updating $P(v)$'s along the way. 
We insert a $0$ before the $i$\textsuperscript{th} position of all the sequences by creating a new superblock, copying the data from the old superblock contained in this leaf to this new superblock in chunks of size $\lceil\lg n\rceil$, and adding $0$'s at appropriate positions when we copy. 
This takes $O(\frac{\lg n}{\lg\lg n} + d) = O(\frac{\lg n}{\lg\lg n})$ time. 
If the size of the new superblock exceeds $2L$, we split it into two superblocks containing roughly the same number of positions. 
The parent of the old leaf becomes the parent, $v$, of both new leaves, and we reconstruct the data structures for $P(v)$ and $R_j(v)$'s in $O(df) = o(\frac{\lg n}{\lg\lg n})$ time. 
This may make a series of internal nodes to overflow, and in the amortized sense, each split of the leaf will only cause a constant number of internal nodes to overflow. 
This gives us an algorithm that supports {\insertop} in $O(\frac{\lg n}{\lg\lg n})$ amortized time. 
The support for {\deleteop} is similar. 

Each {\insertop} or {\deleteop} changes $n$ by $1$. 
This might change the value $\lceil\lg n\rceil$, which will in turn affect $L$, the size of blocks, and the content of $A_1$. 
As $w = \Theta(\lg n)$, $L$ and the block size will only change by a constant factor. 
Thus if we do not change these parameters, all our previous space and time analysis still applies. 
The $o(n)$ time required to reconstruct $A_1$ each time $\lceil\lg n\rceil$ changes can be charged to at least $\Theta(n)$ {\insertop} or {\deleteop} operations.\qed
\end{proof}

As we use a B-tree in our solution, a new problem is to deamortize the process of supporting {\insertop} and {\deleteop}. 
We also need to generalize our results to the case in which the word size of the RAM is $w=\Omega(\lg n)$. 
This leads to the following lemma which presents our solution to the CSPSI problem, and we give a sketch proof (See Appendix~\ref{app:cspsi} for the full proof):  
\begin{lemma}
\label{lem:cspsi}
Consider a collection, $C$, of $d$ sequences of $n$ nonnegative integers each ($d = O(\lg^{\eta} n)$ for any constant $0<\eta<1$), in which each integer requires $k$ bits. 
Under the word RAM model with word size $\Omega(\lg n)$, $C$ can be represented using $O(kdn+w)$ bits to support $\sumop$, $\searchop$, $\updateop$, $\insertop$ and $\deleteop$ in $O(\frac{\lg n}{\lg\lg n})$ time with a buffer of $O(n\lg n)$ bits (for the operation $\updateop(C, j, i, \delta)$, we assume $|\delta| \le \lg n$). 
\end{lemma}
\begin{proof}[sketch]
To deamortize the algorithm for {\insertop} and {\deleteop}, we first observe that the table $A_1$ can be built incrementally each time we perform {\insertop} and {\deleteop}. 
Thus the challenging part is to re-balance the B-tree (i.e. to merge and split its leaves and internal nodes) after insertion and deletion. 
For this we use the {\em global rebuilding} approach of Overmars and van Leeuwen~\cite{ol1981}. 
By their approach, if there exist two constant numbers $c_1 > 0$ and $0<c_2<1$ such that after performing $c_1 n$ insertions and/or $c_2 n$ deletions without re-balancing the B-tree, we can still perform query operations in $O(\frac{\lg n}{\lg\lg n})$ time, and if the B-tree can be rebuilt in $O(f(n)\times n)$ time, we can support insertion or deletion in $O(\frac{\lg n}{\lg\lg n}+f(n))$ worse-case time using additional space proportional to the size of our original data structures and a buffer of size $O(n \lg n)$ bits. 
We first note that if we do not re-balance the B-tree after performing {\deleteop} $c_2 n$ times for any $0 < c_2 < 1$, the time required to answer a query will not change asymptotically. 
This is however different for {\insertop}: one bad case is to perform insertion $\Theta(n)$ times in the same leaf. 
To address this problem, we use the approach of Fleischer~\cite{f1996} as in \cite{sn2010}. 
Essentially, in his approach, at most one internal node and one leaf is split after each insertion, which guarantees that the maximum degree of internal nodes will not exceed $4f$. 
Increasing the maximum degree of internal nodes to $4f$ will not affect our analysis. 
This way after $\Theta(n)$ insertions, query operations can still be performed in $O(\frac{\lg n}{\lg\lg n})$ time. 
Finally, we note that it takes $O(nd)$ time to construct the B-tree, so we can support {\insertop} and {\deleteop} in $O(d+\frac{\lg n}{\lg\lg n}) = O(\frac{\lg n}{\lg\lg n})$ time. 

To apply the global rebuilding approach to our data structure, when the number of {\insertop} and {\deleteop} operations performed exceeds half the initial length of the sequences stored in the data structure, we build a new data structure incrementally. In this data structure, the value of $\lceil\lg n\rceil$ is determined by the value of $n$ when we start the rebuilding process. Using this we can handle the change of the value of $\lceil\lg n\rceil$ smoothly. 
To reduce the space overhead when $w = \omega(\lg n)$, we simply allocate a memory block whose size is sufficient for the new structure until another structure has to be built, and this only increases the space cost of our data structures by a constant factor. 
Thus we can still use pointers of size $O(\lg n)$ bits (not $O(w)$ bits), and a constant number of machine words of $O(w)$ bits are required to record the address of each memory block allocated.\qed
\end{proof}

Recall that in Section~\ref{sec:cspsi}, we stated that Gonz\'{a}lez and Navarro~\cite{gn2009} designed a data structure of $kdn (1 + O(\frac{1}{\sqrt{\lg n}} + \frac{d}{\lg n}))$ bits. 
This is more compact, but it only works for the special case in which $w=\Theta(\lg n)$. 
Gonz\'{a}lez and Navarro~\cite{gn2009} actually stated that their result (See Section~\ref{sec:cspsi}) would work when $w=\Omega(\lg n)$. 
This requires greater care than given in their paper. 
Their strategy is to adapt the approach of M{\"a}kinen and Navarro~\cite{mn2008} developed originally for a dynamic bit vector structure. 
To use it for the CSPSI problem, they split each sequence into three subsequences. 
The split points are the same for all the sequences in $C$. 
The set of left, middle, and right subsequences constitute three collections, and they build CSPSI structures for each of them. 
For each insertion and deletion, a constant number of elements is moved from one collection to another, which will eventually achieve the desired result with other techniques. 
Moving one element from one collection to another means that the first or the last integers of all the subsequences in one collection must be moved to the subsequences in another collection. 
However, their CSPSI structure only supports insertions and deletions of $0$'s at the same position in all subsequences in $O(\lg n)$ time, which means moving one element using their structure cannot be supported fast enough. 
Thus their structure only works when $w=\Theta(\lg n)$. 
Their result can be generalized to the case in which $w=\Omega(\lg n)$ using the approach in Lemma~\ref{lem:cspsi}, but the space will be increased to $O(kdn+w)$ bits and a buffer will be required.

\section{Strings over Small Alphabets}
\label{sec:smallalphabet}
In this section, we consider representing a dynamic string $S[1..n]$ over an alphabet of size $\sigma = O(\sqrt{\lg n})$ to support {\accessop}, {\rankop}, {\selop}, {\insertop} and {\deleteop}. 

\label{sec:stringds}

\paragraph{Data Structures.} Our main data structure is a B-tree constructed over $S$. 
We again let $L = \lceil\frac{\lceil\lg n\rceil^2}{\lg\lceil\lg n\rceil}\rceil$. Each leaf of this B-tree contains a {\em superblock} that has at most $2L$ bits. 
We say that a superblock is {\em skinny} if it has fewer than $L$ bits. 
The string $S$ is initially partitioned into substrings, and each substring is stored in a superblock. 
We number the superblocks consecutively from left to right starting from $1$. 
Superblock $i$ stores the $i$\textsuperscript{th} substring from left to right. 
To bound the number of leaves and superblocks, we require that there do not exist two consecutive skinny superblocks. 
Thus there are $O(\frac{n\lg\sigma}{L})$ superblocks. 

Let $b = \sqrt{\lg n}$, and we require that the degree of each internal node of the B-tree is at least $b$ and at most $2b$. 
For each internal node $v$, we store the following data structures encoded by Lemma~\ref{lem:smallpartial} (let $h$ be the number of children of $v$): 
\begin{itemize}
\item A sequence $U(v)[1..h]$, in which $U(v)[i]$ is the number of superblocks contained in the leaves of the subtree rooted at the $i$\textsuperscript{th} child of $v$; 
\item A sequence $I(v)[1..h]$, in which $I(v)[i]$ stores the number of characters stored 
in the leaves of the subtree rooted at the $i$\textsuperscript{th} child of $v$. 
\end{itemize}

As in Section~\ref{sec:cspsids}, each superblock is further stored in a list of blocks of $\lceil\lceil \lg n\rceil^{3/2}\rceil$ bits each, and only the last block in each list can have free space. 

Finally for each character $\alpha$, we construct an integer sequence $E_{\alpha}[1..t]$ in which $E_{\alpha}[i]$ stores the number of occurrences of character $\alpha$ in superblock $i$. 
We create $\sigma$ integer sequences in this way, and we construct a CSPSI structure, $E$, for them using Lemma~\ref{lem:cspsi}. 

The space analysis is similar to that in Section~\ref{sec:cspsids}, and we have the following lemma when $w = \Theta(\lg n)$ (See Appendix~\ref{app:stringsize} for its proof): 
\begin{lemma}
\label{lem:stringsize}
The above data structures occupy $n\lg \sigma + O(\frac{n\lg\sigma\lg\lg n}{\sqrt{\lg n}})$ bits. 
\end{lemma}



\paragraph{Supporting $\accessop$, $\rankop$ and $\selop$.} We first show how to use our data structures to support query operations. 
\begin{lemma}
\label{lem:stringquery}
The data structures in this section can support $\accessop$, $\rankop$ and $\selop$ in $O(\frac{\lg n}{\lg\lg n})$ time with an additional universal table of $O(\sqrt{n}\polylog (n))$ bits. 
\end{lemma}
\begin{proof}
To support $\accessop(S, i)$, we perform a top-down traversal in the B-tree to find the leaf containing $S[i]$. 
During this traversal, at each internal node $v$, we perform $\searchop$ on $I(v)$ to decide which child to traverse, and perform $\sumop$ on $I(v)$ to update the value $i$. 
When we find the leaf, we follow  
the pointers to find the block containing the position we are looking for, and then retrieve the corresponding character in constant time. 
Thus $\accessop$ takes $O(\frac{\lg n}{\lg\lg n})$ time. 

To compute $\rankop_{\alpha}(S, i)$, we first locate the leaf containing position $i$ using the same process for $\accessop$. 
Let $j$ be the number of the superblock contained in this leaf, which can be computed using $U(v)$ during the top-down traversal. 
Then $\sumop(E, \alpha, j-1)$ is the number of occurrences of $\alpha$ in superblocks $1, 2, \cdots j-1$, which can be computed in $O(\frac{\lg n}{\lg\lg n})$ time by Lemma~\ref{lem:cspsi}. 
To compute the number of occurrences of $\alpha$ up to position $i$ inside superblock $j$, we read the content of this superblock in chunks of size $\lceil\frac{1}{2}\lg n \rceil$ bits. 
As with the support for $\sumop$ in the proof of Lemma~\ref{lem:cspsisum}, this can be done in $O(\frac{\lg n}{\lg\lg n})$ time using a precomputed two-dimensional table $A_2$ of $O(\sqrt{n}\polylog (n))$ bits. 
Our support for $\selop_{\alpha}(S, i)$ is similar (See Appendix~\ref{app:stringselect}). \qed
\end{proof}

\paragraph{Supporting $\insertop$ and $\deleteop$.} To support $\insertop$ and $\deleteop$, we first show how to support them when $w = \Theta(\lg n)$. 
Careful tuning of the techniques for supporting insertions and deletions for the CSPSI problem yields the following lemma, whose proof is in Appendix~\ref{app:stringdelete}: 

\begin{lemma}
\label{lem:stringtheta}
When $w = \Theta(\lg n)$, the data structures in this section can support {\insertop} and {\deleteop} in $O(\frac{\lg n}{\lg\lg n})$ amortized time. 
\end{lemma}


To deamortize the support for {\insertop} and {\deleteop}, we cannot directly use the global rebuilding approach of Overmars and van Leeuwen~\cite{ol1981} here, since we do not want to increase the space cost of our data structures by a constant factor, and the use of a buffer is also unacceptable. 
Instead, we design an approach that deamortizes the support for {\insertop} and {\deleteop} completely, and differently from the original global rebuilding approach of Overmars and van Leeuwen~\cite{ol1981}, our approach neither increases the space bounds asymptotically, nor requires any buffer. 
We thus call our approach {\em succinct global rebuilding}. 
This approach still requires us to modify the algorithms for {\insertop} and {\deleteop} so that after $c_1 n$ insertions ($c_1 > 0$) and $c_2 n$ deletions ($0<c_2<1$), a query operation can still be supported in $O(\frac{\lg n}{\lg\lg n})$ time. 
We also start the rebuilding process when the number of {\insertop} and {\deleteop} operations performed exceeds half the initial length of the string stored in the data structure. 
The main difference between our approach and the original approach in \cite{ol1981} is that during the process of rebuilding, we never store two copies of the same data, i.e. the string $S$. 
Instead, our new structure stores a prefix, $S_p$, of $S$, and the old data structure stores a suffix, $S_s$, of $S$. 
During the rebuilding process, each time we perform an insertion or deletion, we perform such an operation on either $S_p$ or $S_s$. 
After that, we remove the first several characters from $S_s$, and append them to $S_p$. 
By choosing parameters and tuning our algorithm carefully, we can finish rebuilding after we perform at most $n_0 / 3$ update operations where $n_0$ is the length of $S$ when we start rebuilding. 
During the rebuilding process, we can use both old and new data structures to answer each query in $O(\frac{\lg n}{\lg\lg n})$ time. 
All the details can be found in Appendix~\ref{app:succinctrebuild}. 

To reduce the space overhead when $w = \omega(\lg n)$, we adapt the approach of M{\"a}kinen and Navarro~\cite{mn2008} (See Appendix~\ref{app:lgnchange}).  
We finally use the approach of Gonz\'{a}lez and Navarro~\cite{gn2009} to compress our representation. 
Since their approach is only applied to the superblocks (i.e. it does not matter what kind of tree structures are used, since additional tree arrangement operations are not required when the number of bits stored in a leaf is increased due to a single update in their solution), we can use it here directly. 
Thus we immediately have: 
\begin{lemma}
\label{lem:smallalphabet}
Under the word RAM model with word size $w=\Omega(\lg n)$, a string $S$ of length $n$ over an alphabet of size $\sigma = O(\sqrt{\lg n})$ can be represented using $nH_0 + O(\frac{n\lg\sigma\lg\lg n}{\sqrt{\lg n}}) + O(w)$ bits to support $\accessop$, $\rankop$, $\selop$, $\insertop$ and $\deleteop$ in $O(\frac{\lg n}{\lg\lg n})$ time. 
\end{lemma}

\section{Strings over General Alphabets}
\label{sec:generalalphabet}

We follow the approach of Ferragina~\etal~\cite{fmmn2007} that uses a generalized wavelet tree to extend results on strings over small alphabets to general alphabets. 
Special care must be taken to avoid increasing the $O(w)$-bit term in Lemma~\ref{lem:smallalphabet} asymptotically. 
We now present our main result (see Appendix~\ref{app:string} for its proof): 
\begin{theorem}
\label{thm:string}
Under the word RAM model with word size $w=\Omega(\lg n)$, a string $S$ of length $n$ over an alphabet of size $\sigma$ can be represented using $nH_0 + O(\frac{n\lg\sigma\lg\lg n}{\sqrt{\lg n}}) \linebreak + O(w)$ bits to support $\accessop$, $\rankop$, $\selop$, $\insertop$ and $\deleteop$ operations in $O(\frac{\lg n}{\lg\lg n}(\frac{\lg \sigma}{\lg\lg n}+1))$ time. 
When $\sigma = O(\polylog(n))$, all the operations can be supported in $O(\frac{\lg n}{\lg\lg n})$ time. 
\end{theorem}

The following corollary is immediate:
\begin{corollary}
Under the word RAM model with word size $w=\Omega(\lg n)$, a bit vector of length $n$ can be represented using $nH_0 + o(n) + O(w)$ bits to support $\accessop$, $\rankop$, $\selop$, $\insertop$ and $\deleteop$ in $O(\frac{\lg n}{\lg\lg n})$ time. 
\end{corollary}


\section{Applications}
\paragraph{Dynamic Text Collection.} M{\"a}kinen and Navarro~\cite{mn2008} showed how to use a dynamic string structure to index a text collection $N$ which is a set of text strings to support string search. 
For this problem, $n$ denotes the length of the text collection $N$ when represented as a single string that is the concatenations of all the text strings in the collection (a separator is inserted between texts). 
A key structure for their solution is a dynamic string. 
Gonz\'{a}lez and Navarro~\cite{gn2009} improved this result in~\cite{mn2008} by improving the string representation. 
If we use our string structure, we directly have the following lemma, which improves the running time of the operations over the data structure in~\cite{gn2009} by a factor of $\lg\lg n$: 
\begin{theorem}
Under the word RAM model with word size $w=\Omega(\lg n)$, a text collection $N$ of size $n$ consisting of $m$ text strings over an alphabet of size $\sigma$ can be represented in $nH_h+o(n)\cdot\lg\sigma+O(m\lg n+w)$ bits\footnote{$H_h$ is the $h$\textsuperscript{th}-order entropy of the text collection when represented as a single string.} for any $h \le (\alpha\log_{\alpha} n) -1$ and any constant $0<\alpha<1$ to support: 
\begin{itemize}
\item the counting of the number of occurrences of a given query substring $P$ in $N$ in $O(\frac{|P|\lg n}{\lg\lg n}(\frac{\lg \sigma}{\lg\lg n}+1))$ time;
\item After counting, the locating of each occurrence in $O(\lg^2 n (\frac{1}{\lg\lg n}+\frac{1}{\lg \sigma}))$ time; 
\item Inserting and deleting a text $T$ in $O(\frac{|T|\lg n}{\lg\lg n}(\frac{\lg \sigma}{\lg\lg n}+1))$ time; 
\item Displaying any substring of length $l$ of any text string in $N$ in $O(\lg^2 n (\frac{1}{\lg\lg n}+\frac{1}{\lg \sigma}) + \frac{l\lg n}{\lg\lg n}(\frac{\lg \sigma}{\lg\lg n}+1))$ time. 
\end{itemize}
\end{theorem}

\paragraph{Compressed Construction of Text Indexes.} Researchers have designed space-efficient text indexes whose space is essentially a compressed version of the given text, but the construction of these text indexes may still require a lot of space. 
M{\"a}kinen and Navarro~\cite{mn2008} showed how to use their dynamic string structure to construct a variant of FM-index~\cite{fmmn2004} using as much space as what is required to encode the index. 
This is useful because it allows text indexes of very large text to be constructed when the memory is limited. 
Their result was improved by Gonz\'{a}lez and Navarro~\cite{gn2009}, and the construction time can be further improved by a factor of $\lg\lg n$ using our structure: 
\begin{theorem}
A variant of a FM-index of a text string $T[1..n]$ over an alphabet of size $\sigma$ can be constructed using $nH_h + o(n)\cdot\lg\sigma$ bits of working space in $O(\frac{n\lg n}{\lg\lg n}(\frac{\lg \sigma}{\lg\lg n}+1))$ time for any $h \le (\alpha\log_{\alpha} n) -1$ and any constant $0<\alpha<1$. 
\end{theorem}

\section{Concluding Remarks}
In this paper, we have designed a succinct representation of dynamic strings that provide more efficient operations than previous results, and we have successfully applied it to improve previous data structures on text indexing. 
As a string structure supporting rank and select is used in the design of succinct representations of many data types, we expect our data structure to play an important role in the future research on succinct dynamic data structures. 
We have also created some new techniques to achieve our results. We particularly think that the approach of succinct global rebuilding is interesting, and expect it to be useful for deamortizing algorithms on other succinct data structures.

\bibliographystyle{splncs}
\bibliography{dynamicstring}

\newpage
\appendix
\begin{center}
  \bf \Large Appendices
\end{center}

\section{The Comparison of Our Results and Previous Results}\label{app:compare}

\begin{table}[h]
\begin{center}
\begin{tabular} {|p{2.3cm}|p{1.6cm}|p{2.3cm}|p{2.7cm}|p{2.7cm}|} \hline
                              &alphabet  &space (bits) &{\accessop}, {\rankop} and {\selop} &{\insertop} and {\deleteop}   \\ \hline
M{\"a}kinen \& Navarro~\cite{mn2008} &General &$nH_0 + o(n)\cdot\lg \sigma$ &$O(\lg n \log_q\sigma)$ ($q = o(\sqrt{\lg n})$) &$O(q\lg n \log_q\sigma)$ \\
Lee \& Park~\cite{lp2009} & &$n\lg\sigma + o(n)\cdot\lg \sigma$ &$O(\lg n (\frac{\lg\sigma}{\lg\lg n} +1))$ &$O(\lg n (\frac{\lg\sigma}{\lg\lg n} +1))$ amortized \\
Gonz\'{a}lez \& Navarro~\cite{gn2009} & &$nH_0 + o(n)\cdot\lg \sigma$ &$O(\lg n (\frac{\lg\sigma}{\lg\lg n} +1))$ &$O(\lg n (\frac{\lg\sigma}{\lg\lg n} +1))$\\
Gupta~\etal~\cite{ghsv2007} & &$n\lg\sigma + \lg\sigma(o(n)+O(1))$ &$O(\lg\lg n)$ &$O(\lg^{\epsilon} n)$ amortized ($0<\epsilon < 1$) \\
This paper & &$n H_0 + o(n)\cdot\lg\sigma$ &$O(\frac{\lg n}{\lg\lg n}(\frac{\lg \sigma}{\lg\lg n}+1))$ &$O(\frac{\lg n}{\lg\lg n}(\frac{\lg \sigma}{\lg\lg n}+1))$\\
\hline

Lee \& Park~\cite{lp2009} &$\polylog(n)$ &$n\lg\sigma + o(n)\cdot\lg \sigma$ &$O(\lg n)$ &$O(\lg n)$ amortized \\
Gonz\'{a}lez \& Navarro~\cite{gn2009} & &$nH_0 + o(n)\cdot\lg \sigma$ &$O(\lg n)$ &$O(\lg n)$\\
Gupta~\etal~\cite{ghsv2007} & &$n\lg\sigma + \lg\sigma(o(n)+O(1))$ &$O(1)$ &$O(\lg^{\epsilon} n)$ amortized ($0<\epsilon < 1$) \\
This paper & &$n H_0 + o(n)\cdot\lg\sigma$ &$O(\frac{\lg n}{\lg\lg n})$ &$O(\frac{\lg n}{\lg\lg n})$\\
\hline

Blandford \& Blelloch~\cite{bb2004}  & Binary &$O(nH_0)$    &$O(\lg n)$ &$O(\lg n)$ \\ 
Chan~\etal~\cite{chl2004} &  &$O(n)$ &$O(\lg n)$ &$O(\lg n)$ \\ 
Chan~\etal~\cite{chls2007} &  &$O(n)$ &$O(\frac{\lg n}{\lg\lg n})$ &$O(\frac{\lg n}{\lg\lg n})$\\
M{\"a}kinen \& Navarro~\cite{mn2008} &  &$nH_0 + o(n)$ &$O(\lg n)$ &$O(\lg n)$ \\
Sadakane \& Navarro~\cite{sn2010} & &$n+o(n)$ &$O(\frac{\lg n}{\lg\lg n})$ &$O(\frac{\lg n}{\lg\lg n})$\\
Gupta~\etal~\cite{ghsv2007} & &$nH_0 + o(n))$ &$O(\lg\lg n)$ &$O(\lg^{\epsilon} n)$ amortized ($0<\epsilon < 1$) \\
This paper & &$nH_0 + o(n)$ &$O(\frac{\lg n}{\lg\lg n})$ &$O(\frac{\lg n}{\lg\lg n})$\\
\hline
\end{tabular}
\end{center}
\caption{A comparison of previous results and our results on succinct representations of dynamic strings. }
\label{table:compare}
\end{table}

\section{Proof of Lemma~\ref{lem:cspsisize}}
\label{app:cspsisize}
It takes $kd n$ bits to encode the integers in $C$. To compute the overhead of storing $C$ in blocks using our strategy, we first bound the space required to store the pointers used in the lists of blocks. 
The number of pointers is linear in the number of blocks which is  $O(\frac{kd n}{\lg^{3/2} n})$, and since it takes $\Theta(\lg n)$ to encode each pointer, all the pointers occupy $O(\frac{kd n}{\sqrt{\lg n}})$ bits in total. 
We then bound the space wasted in the blocks that are partially full. 
As only one block in each superblock is allowed to have some free space, the number of such blocks is at most the number of superblocks which is $O(\frac{kd n \lg\lg n}{\lg^2 n})$. 
Since up to $\lceil\lceil \lg n\rceil^{3/2}\rceil$ bits can be wasted for each such block, the total number of wasted bits here is at most $O(\frac{kd n \lg\lg n}{\sqrt{\lg n}})$. 
Thus to store $C$ using superblocks and blocks (without considering the B-tree), it requires $kd n + O(\frac{kd n \lg\lg n}{\sqrt{\lg n}})$ bits. 

It now suffices to compute the additional space required to store the B-tree. 
The number of nodes in the B-tree is linear in the number of superblocks which is $O(\frac{kd n \lg\lg n}{\lg^2 n})$. 
Thus the structure of the B-tree (which is an ordinal tree without considering any additional information stored with each node) can be stored in $O(\frac{kd n \lg\lg n}{\lg n})$ bits using a constant number of pointers per node. 
By Lemma~\ref{lem:smallpartial}, the $d+1$ sequences for each internal node can be maintained using $O(df\lg n)$ bits with a precomputed universal table of size $o(n)$ bits (we only need one copy of this {\em universal} table for the sequences for all the internal nodes). 
Thus the B-tree can be stored using $O(\frac{kd^2f n \lg\lg n}{\lg n})+o(n)$ bits excluding the data stored in its leaves. 
Therefore, the total space cost is $kd n (1+ O(\frac{\lg\lg n}{\sqrt{\lg n}}) + O(\frac{df \lg\lg n}{\lg n}))$ bits. 
Since we assume $d = O(\lg^{\eta} n)$ and $f = \lg^{\lambda} n$, the above space cost is  $kd n (1+ O(\frac{\lg\lg n}{\sqrt{\lg n}}) +  O(\frac{\lg\lg n}{\lg^ {1-\eta - \lambda}n}))$. 
When $0<\lambda < 1-\eta$, the space cost becomes $kd (n+o(n))$ bits. \qed


\section{Details of Supporting $\searchop$ and $\updateop$ in the Proof of Lemma~\ref{lem:cspsisum}}
\label{app:cspsisearch}
The support for operation $\searchop(C, j, x)$ is similar to that for $\sumop$. 
Initially we let $v$ be the root of the tree, and set $r$ to be $0$. 
We again start at the root $v$, and by computing $c = \searchop(R_j(v), x)$, we know that the $c$\textsuperscript{th} child of the root contains the result. 
We also increase the value of $r$ by $\sumop(P(v), c-1)$. 
We then set $v$ to be its $c$\textsuperscript{th} child, decrease the value of $x$ by $\sumop(R_j(v), c-1)$, and the process continues until we reach a leaf. 
We then process the corresponding superblock in chunks of size $\lceil \frac{1}{2}\lg n \rceil$, summing up the integers from $Q_j$ in this superblock using $A_1$, until we get a sum that is larger than or equal to the current value $x$. 
A binary search in $O(\lg\lg n)$ time in the last chunk we read, with the help of table $A_1$, will give us the result (we also need to increase the result of the binary search by the value stored in $r$ when we return it as the result). 
The entire process takes $O(\frac{\lg n}{\lg\lg n})$ time. 

To support $\updateop(C, j, i, \delta)$, we perform a top-down traversal as we do for $\sumop$ until we reach a leaf. 
During the traversal, each time we go from a node $v$ to its child (let $c$ be the rank of this child among its sibling), we update $R_j(v)$ by performing $\updateop(R_j(v), c, \delta)$. 
When we reach a leaf, we can locate the $k$ bits storing $Q_j[i]$ in $O(\frac{\sqrt{\lg n}}{\lg\lg n})$ time, as we only have to follow the pointers between the blocks of the superblock $O(\frac{\sqrt{\lg n}}{\lg\lg n})$ times. 
This will allow us to update $Q_j[i]$. 
The entire process takes $O(\frac{\lg n}{\lg\lg n})$ time.\qed

\section{Proof of Lemma~\ref{lem:cspsiinsert}}
\label{app:cspsiinsert}
The operations {\insertop} and {\deleteop} change the size of $C$ by increasing and decreasing $n$ by $1$, respectively. 
When $n$ changes, sometimes the value $\lceil\lg n\rceil$ also changes, and this affects our data structures: First, this changes the maximum and minimum sizes of superblocks and the size of blocks. 
Second, since we use a precomputed universal table $A_1$ to process $\lceil\frac{1}{2}\lg n\rceil$ bits in chunks, whose content depends on $\lceil\frac{1}{2}\lg n\rceil$, $A_1$ may change when $\lceil\lg n\rceil$ changes. 
Thus we need to handle the case in which $\lceil\lg n\rceil$ changes after an  {\insertop} or {\deleteop} operation. 

We first consider the case in which $\lceil\lg n\rceil$ does not change after we perform $\insertop$ or $\deleteop$. 
To support $\insertop(C, i)$, we start from the root and traverse down the B-tree as we do for $\updateop$ until we reach a leaf. 
During the traversal, each time we go from a node $v$ to its child (let $c$ be the rank of this child among its sibling), we update $P(v)$ by performing $\updateop(P(v), c, 1)$. 
When we reach a leaf, we insert a $0$ before the $i$\textsuperscript{th} position of all the sequences by creating a new superblock, copying the data from the old superblock contained in this leaf to this new superblock, and adding $0$'s at appropriate positions when we copy. 
We then replace the old superblock by the new superblock, and deallocate the memory for the old superblock. 
As we can copy the bits from the old superblock to the new superblock in chunks of size $\lceil\lg n\rceil$, and it takes constant time to add a $0$ into a sequence, the copying process takes $O(\frac{\lg n}{\lg\lg n} + d) = O(\frac{\lg n}{\lg\lg n})$ time. 
Since each superblock has $O(\frac{\sqrt{\lg n}}{\lg\lg n})$ blocks, the deallocation and allocation of a superblock takes $O(\frac{\sqrt{\lg n}}{\lg\lg n})$ time. 
Combined with the $O(\frac{\lg n}{\lg\lg n})$ time required to traverse down the B-tree, our algorithm takes $O(\frac{\lg n}{\lg\lg n})$ time so far. 

After the above process, there are two cases. First, the size of the new superblock does not exceed $2L$. In this case, we do nothing. 
Second, the size of the new superblock has more than $2L$ bits, which means the leaf has $m = \lfloor\frac{2L}{dk}\rfloor + 1$ integers from each sequence in $C$. 
In this case, we split the leaf into two. The left new leaf stores first $\lceil m/2\rceil$ integers from each sequence, and the right one stores the rest. 
Clearly the size of the superblock for either leaf is between $L/2$ and $2L$. 
This requires another copying process, similar to that in the previous paragraph, which takes $O(\frac{\lg n}{\lg\lg n})$ time. 
The parent of the old leaf becomes the parent, $v$, of both new leaves. 
We then reconstruct the data structures for $P(v)$ and $R_j(v)$'s. 
By Lemma~\ref{lem:smallpartial}, this requires $O(df) = O(\lg ^{\eta+\lambda} n) = o(\frac{\lg n}{\lg\lg n})$ time (recall that we have $0<\lambda<1-\eta$ by Lemma~\ref{lem:cspsisize}). 
However, the parent might overflow (i.e. have more than $2f$ children), and if we split the parent into two nodes, this might in turn cause more nodes to overflow along the path to the root. 
Thus we need to split all the $O(\frac{\lg n}{\lg\lg n})$ ancestors of the new leaves and rebuild their associated data structures in the worst case. 
It is well-known that in the amortized sense, each split of the leaf of a B-tree will only cause a constant number of internal nodes to overflow. 
This means that each {\insertop} requires $o(\frac{\lg n}{\lg\lg n})$ amortized time to rebuild the data structures for internal nodes. 
As a result, we now have an algorithm that can support {\insertop} in $O(\frac{\lg n}{\lg\lg n})$ amortized time. 


Our support for {\deleteop} is analogous to our support for {\insertop}, which takes $O(\frac{\lg n}{\lg\lg n})$ amortized time. 

We now need only handle the case in which $\lceil\lg n\rceil$ is increased or decreased by $1$ after an  {\insertop} or {\deleteop} operation. 
First, this change will cause the value $L$, as well as the size of blocks, to change. 
With the assumption, $w = \Theta(\lg n)$ in this lemma, no matter how many times we change the value of $\lceil\lg n\rceil$, $L$ will only change by a constant factor. 
The same applies to the value we choose as block size. 
Thus, when $\lceil\lg n\rceil$ changes, if we do not change the block size, or the maximum and minimum sizes for superblocks, it is easy to verify that all our previous space and time analysis still applies. 

We still need to take care of the table $A_1$ built in Lemma~\ref{lem:cspsisum}. 
Recall that $A_1$ has $O((\sqrt{n}\lg^2 n)/k)$ bits. 
Thus if we do not update $A_1$ (i.e. if we keep use the table built for the original given collection), its size may not always be a lower order term. 
To address this problem, an easy strategy is to rebuild $A_1$ each time $\lceil \lg n\rceil$ is increased or decreased by $1$. 
This take $o(n)$ time, but as this happens only when we have performed {\insertop} or {\deleteop} at least $\Theta(n)$ times, we can charge this cost to $\Theta(n)$  {\insertop} or {\deleteop} operations. 
Hence {\insertop} and {\deleteop} can be supported in $O(\frac{\lg n}{\lg\lg n})$ amortized time.\qed

\section{Proof of Lemma~\ref{lem:cspsi}}
\label{app:cspsi}
In this proof, we first deamortize the support for {\insertop} and {\deleteop} of Lemma~\ref{lem:cspsiinsert}, and then generalize our results to the case in which the word size of the RAM is $w=\Omega(\lg n)$. 

To deamortize the algorithm for {\insertop} and {\deleteop}, we first deamortize the process of rebuilding table $A_1$. 
As the content of $A_1$ only depends on $n$, we can simply construct the new tables incrementally each time we {\insertop} or {\deleteop}. 
The same strategy can be used for the table constructed when we use Lemma~\ref{lem:smallpartial} to encode all the $P(v)$'s and $R_j(v)$'s. 

The challenging part of this proof is to re-balance the B-tree (i.e. to merge and split its leaves and internal nodes) after insertion and deletion. 
For this we use the {\em global rebuilding} approach of Overmars and van Leeuwen~\cite{ol1981}. 
By their approach, if there exist two constant numbers $c_1 > 0$ and $0<c_2<1$ such that after performing $c_1 n$ insertions and/or $c_2 n$ deletions without re-balancing the B-tree, we can still perform query operations in $O(\frac{\lg n}{\lg\lg n})$ time, and if the B-tree can be rebuilt in $O(f(n)\times n)$ time, we can support insertion or deletion in $O(\frac{\lg n}{\lg\lg n}+f(n))$ worse-case time using additional space proportional to the size of our original data structures and a buffer of size $O(n\lg n)$ bits. 
We first note that if we do not re-balance the B-tree after performing {\deleteop} $c_2n$ times for any $0 < c_2 < 1$, the time required to answer a query will not change asymptotically. 
This is however different for {\insertop}: one bad case is to perform insertion $\Theta(n)$ times in the same leaf. 

To address the problem related to {\insertop}, we use the approach of Fleischer~\cite{f1996} as in \cite{sn2010}\footnote{The details of using Fleischer's approach in \cite{sn2010} were omitted in the original paper and were given in private communication with Gonzalo Navarro. They also used the technique of global rebuilding in \cite{sn2010}, but we are not convinced of the correctness of their strategy, so we use global rebuilding differently here.}. 
Fleischer originally showed how to maintain a $(a, 2b)$ tree where $b \ge 2a$ such that insertions and deletions can be performed in $O(1)$ worst-case time after the positions in the leaves where such update operations occur are located. 
Essentially, in his approach, at most one internal node is split after each insertion, which guarantees that the maximum degree of internal nodes will not exceed $2b$. 
This is done by maintaining a pointer called {\em $r$-pointer} for each leaf which initially points to the parent of this leaf. 
If we insert the new key into leaf $B$, and if the $r$-pointer, $r_B$, of $B$ points to node $v$, we check if $v$ has more than $b$ children. 
If it has, we split it into two smaller nodes. 
If now the parent, $w$, of $v$ has more than $b$ children, we mark the edges to its two new children as connected pair, and this information will be used when we split $w$ in the future. 
If we find that $v$ has less than $b$ children when we check its size, we either split the leaf $B$ if $v$ is the root of the tree, or we make $r_B$ point to the parent of $v$. 
To apply the above approach to our data structures, note that in the above process, each time a key value of $\lceil\lg n\rceil$ bits is inserted into a leaf, while in our problem, each time a character which occupies $\lceil\lg \sigma\rceil$ bits is inserted. 
Thus if we move the pointer $r_B$ of any leaf $B$ after $\lceil\lg n\rceil/\lceil\lg \sigma\rceil$ characters are inserted into it (assume that $\lceil\lg n\rceil$ is divisible by $\lceil\lg \sigma\rceil$ for simplicity), we can use Fleischer's approach here. 
The information about connected edges can be stored using a bitmap of size $4f$ for each internal node. 
Using this approach, the maximum degree of internal nodes is $4f$, and our previous analysis still applies. 
This way after $O(n)$ insertions, query operations can still be performed in $O(\frac{\lg n}{\lg\lg n})$ time. 
Finally, we note that it takes $O(nd)$ time to construct the B-tree, so we can support {\insertop} and {\deleteop} in $O(d+\frac{\lg n}{\lg\lg n}) = O(\frac{\lg n}{\lg\lg n})$ time. 

To apply the global rebuilding approach to our data structure, when the number of {\insertop} and {\deleteop} operations performed exceeds half the initial length of the sequences stored in the data structure, we build a new data structure incrementally. In this data structure, the value of $\lceil\lg n\rceil$ is determined by the value, $n_0$, of $n$ when we start the rebuilding process. 
After we finish rebuilding, the value of $n$ can only be changed by a constant factor, thus we can still use $\lceil\lg n_0\rceil$ as the value of $\lceil\lg n\rceil$ without affecting time or space bounds. 
Using this we can handle the change of the value of $\lceil\lg n\rceil$ smoothly, since the difference between $\lceil\lg n\rceil$ and $\lceil\lg n_0\rceil$ is at most $1$ before we start rebuilding again.



Extending our result to the more general case in which $w = \Omega(\lg n)$ is more difficult than the static case. 
For a static data structure, since its size does not change, we can store it continuously in the memory, and thus the pointer size only depends on the size of the data structure, which is independent of $w$. 
Therefore, space bounds of static data structures designed for the case in which $w=\Theta(\lg n)$ usually remain the same when $w = \Omega(\lg n)$. 
However, it is not the same for a typical dynamic data structure. 
Since the size of dynamic structures changes, their components are stored in different memory blocks, and thus pointers of size $w$ are needed to record addresses in memory. 
When $n$ is small enough that $w = \omega(\lg n)$, these pointers may occupy too much space.

To reduce the space overhead when $w = \omega(\lg n)$, we first store each internal node of the B-tree using the same number of bits. 
This will increase the space cost of the data structures for internal nodes by a constant factor. 
For each value of $\lceil\lg n\rceil$, we can compute the maximum number of bits required to store all these internal nodes when the number of internal nodes in the B-tree is maximized, and allocate a consecutive memory area that is just big enough for them. 
We again waste a constant fraction of space in this memory area. 
We divide this area into segments of the same size, and each segment is just big enough to store an internal node. 
To record the starting address of this area, we need $w$ bits, but to locate any segment inside this area, pointers of size $O(\lg n)$ bits are enough. 
To encode the B-tree, we also need to encode the pointers between parents and children, and $O(\lg n)$ bits are enough for each such pointer (the pointers that point to leaves can also be encoded in  $O(\lg n)$ bits, and we will show how to achieve this later). 
We also store these pointers in internal nodes, so that each internal node has a pointer for its parent and $O(\sqrt{\lg n})$ pointers for its children. 
To handle the insertion and deletion of internal nodes, we chain the empty segments by wasting one pointer in each segment. 
The leaves are maintained in the same manner, and they can be referred to using $O(\lg n)$ bits. 
This increases the total space cost to $O(kdn + w)$ bits. \qed


\section{Proof of Lemma~\ref{lem:stringsize}}
\label{app:stringsize}

The string occupies $n\lg\sigma$ bits. 
The space required for the pointers that link the blocks is $O(\frac{n\lg\sigma}{\sqrt{\lg n}})$ bits, and the space wasted in the partially full blocks is $O(\frac{n\lg\sigma\lg\lg n}{\sqrt{\lg n}})$ bits. 
The B-tree has $O(\frac{n\lg\sigma\lg\lg n}{\lg ^2 n})$ nodes, and each internal nodes encodes data of $O(\lg^{3/2} n)$ bits (including $\Theta(\sqrt{\lg n})$ pointers to its children and parent), so the internal nodes of the B-tree require  $O(\frac{n\lg\sigma\lg\lg n}{\sqrt{\lg n}})$ bits. 
Finally, the CSPSI structure $E$ occupies $O(\frac{n\lg\sigma}{L}\times \sigma \times \lg n) = O(\frac{n\lg\sigma\lg\lg n}{\sqrt{\lg n}})$ bits, and its buffer requires $O(\frac{n\lg\sigma}{L}\times \lg(\frac{n\lg\sigma}{L})) = O(\frac{n\lg\sigma\lg\lg n}{\lg n})$ bits. 
Therefore, all the data structures occupy $n\lg \sigma + O(\frac{n\lg\sigma\lg\lg n}{\sqrt{\lg n}})$ bits, including the precomputed tables for the $U(v)$'s and $I(v)$'s.\qed

\section{Support for $\selop$ in the Proof of Lemma~\ref{lem:stringquery}}
\label{app:stringselect}

Our algorithm for $\selop_{\alpha}(S, i)$ first computes the number of the superblock containing the $i$\textsuperscript{th} occurrence of $\alpha$ in $S$, which is $j = \searchop(E, \alpha, i)$. 
We also know that the $i$\textsuperscript{th} occurrence of $\alpha$ in $S$ is the $(i-\sumop(E, \alpha, i-1))$\textsuperscript{th} occurrence of $\alpha$ in superblock $j$. 
We then locate superblock $j$ by traversing down the B-tree, using $\sumop$ and $\searchop$ operations on the $U(v)$'s. 
Once we find superblock $j$, we read its content in chunks of size $\lceil\frac{1}{2}\lg n \rceil$. 
With table $A_2$, we can find the chunk containing the occurrence we are looking for in $O(\frac{\lg n}{\lg\lg n})$ time. 
A binary search within this chunk in $O(\lg\lg n)$ time using $A_2$ will find this occurrence. 
Thus $\selop$ can be supported in $O(\frac{\lg n}{\lg\lg n})$ time. \qed

\section{Proof of Lemma~\ref{lem:stringtheta}}
\label{app:stringdelete}

Operations $\insertop$ and $\deleteop$ can possibly change the value of $\lceil \lg n \rceil$. 
We first consider the case in which the value $\lceil \lg n \rceil$ remains unchanged after an $\insertop$ or $\deleteop$ operation. 

To support $\insertop_{\alpha}(S, i)$, we first locate the leaf containing $S[i-1]$ using the same process for $\accessop$. 
Let $j$ be the number of the superblock contained in this leaf, and assume that $S[i-1]$ is the $z$\textsuperscript{th} character stored in this leaf. 
We then insert $\alpha$ after this character, and shift the remaining characters, starting from the character that is currently the $(z+1)$\textsuperscript{th} character in superblock $j$, in chunks of size $\lceil\frac{1}{2}\lg n \rceil$ bits. 
If the last block of the superblock does not have enough free space for one character before insertion, but the superblock has no more than $2L-\lceil\lceil \lg n\rceil^{3/2}\rceil$ bits, we create another block and append it to the list of blocks for this superblock so that the insertion can be performed. 
After this, we update $E$ by calling $\updateop(E, \alpha, j, 1)$, and we also visit the ancestors of this leaf, updating their $I(v)$ sequences by incrementing a certain number in the sequence by $1$. 
We then terminate the process, which takes $O(\frac{\lg n}{\lg\lg n})$ time. 

If superblock $j$ is so full that the insertion of a single character cannot be done, we check superblock $j-1$. 
If superblock $j-1$ is not full, we remove the first character from superblock $j$, and insert it after the last character in superblock $j-1$. 
We then shift the second, third, until the $(z-1)$\textsuperscript{th} character to the left by one position, again in chunks of size $\lceil\frac{1}{2}\lg n \rceil$. 
We then insert $\alpha$ at the $z$\textsuperscript{th} position inside superblock $j$. 
Then we update the $I(v)$ sequences of the ancestors of the leaf containing superblock $j-1$, perform three updates to $E$, and then terminate. 
This process takes $O(\frac{\lg n}{\lg\lg n})$ time. 

Finally we need to consider the case in which superblock $j-1$ does not exist, or it is full. 
In this case, we split superblock $j$ into two new superblocks, and the left one has only  one character (i.e. the first character stored in superblock $j$). 
Then the second new superblock is not full, so that we can insert $\alpha$ after the $(z-1)$\textsuperscript{th} character in it. 
Both new leaves will be the children of the parent, $v$, of the original superblock $j$, and we need to reconstruct the data structures $U(v)$ and $I(v)$ in $O(\sqrt{\lg n})$ time. 
We then modify the $U(u)$ and $I(u)$ sequences of each ancestor, $u$, of $v$, using the $\updateop$ operator (we pay $O(1)$ time for every node $u$ here as we only need to increase one integer in $U(u)$ by $1$ and another integer in $I(u)$ by $1$). 
Finally we also need to update $E$. Note that the new superblocks created are numbered $j$ and $j+ 1$. 
For the new superblock $j$, we first call $\insertop(E, j)$ and then call $\updateop(E, \beta, j, 1)$, where $\beta$ is the only character in this superblock. 
For superblock $j+1$, performing $\updateop(E, \beta, j+1, -1)$ and $\updateop(E, \alpha, j+1, 1)$ will reflect the removal of one character $\beta$ from this superblock and the insertion of one character $\alpha$ into this block. 

So far the only problem we have not considered is the possibility that splitting a leaf may incur a series of splits of the nodes along the path from this leaf to the root. 
Since the number of splits incurred after one insertion is $O(1)$ in the amortized sense, we can support {\insertop} in $O(\frac{\lg n}{\lg\lg n})$ amortized time. 

To support $\deleteop(S, i)$, 
we first locate the leaf, j, containing $S[i]$ in $O(\frac{\lg n}{\lg\lg n})$ time. 
This process also tells us the position of $S[i]$ in superblock $j$ (assume that $S[i]$ is the $q$\textsuperscript{th} character in superblock $j$). 
This allows us to retrieve the character, $\alpha$, stored in $S[i]$. 
We remove this character from its superblock by shifting, and if this makes the last block of superblock $j$ empty, we simply remove it. 
After this, there are three cases. In the first case, superblock $j$ becomes empty. 
In this case, we remove the leaf containing superblock $j$, and rebuild the data structures stored in the parent of the leaf containing superblock $j$. 
For any other ancestor, $v$, of this leaf, we update the data structures $I(v)$ and $U(v)$ in constant time to reflect the fact that there is one less character $\alpha$ and one less superblock stored in the subtree rooted at $v$. 
We also perform $\deleteop(E, j)$ in $O(\frac{\lg n}{\lg\lg n})$ time. 

In the second case, superblock $j$ is not empty, and the fact that it has one less character does not violate the restriction that there are no two adjacent skinny superblocks. 
This can happen if after removing $S[j]$, superblock $j$ is not skinny, or it is skinny, but neither of the two superblocks adjacent to it is skinny. 
In this case, only an $\updateop$ operation need be performed on $I(v)$ for each ancestor, $v$, of the leaf containing superblock $j$. 
We also perform $\updateop(E, \alpha, j, -1)$  in $O(\frac{\lg n}{\lg\lg n})$ time. 

In the third case, superblock $j$ is not empty, but both superblock $j$ and at least one of its adjacent superblocks are skinny. 
Without loss of generality, assume that superblock $j-1$ is skinny. 
In this case, we locate superblock $j-1$ by performing a top-down traversal, remove its last character, and insert it to the first position in superblock $j$, so that superblock $j$ is no longer skinny. 
If superblock $j-1$ becomes empty, we remove it. 
Updates to $E$ and the ancestors of superblocks $j-1$ and $j$ are performed in a similar manner as in the first two cases. 

It is clear that our algorithm for each case takes $O(\frac{\lg n}{\lg\lg n})$ time. 
We also need to consider the possibility of causing internal nodes to underflow after we remove an empty superblock. 
Since each deletion causes $O(1)$ amortized number of internal nodes to underflow, $\deleteop$ can be supported in $O(\frac{\lg n}{\lg\lg n})$ amortized time.

We complete this proof by pointing out that the same approach used in the proof of Lemma~\ref{lem:cspsiinsert} to handle the changes of $\lceil\lg n\rceil$ can also be used here.\qed

\section{Deamortization of the Support for {\insertop} and {\deleteop} Operations over Strings}
\label{app:succinctrebuild}

We first modify the support for {\insertop} so that it takes $O(\frac{\lg n}{\lg\lg n})$ time, and that for any constant $c_1>1$, after performing $c_1 n$ insertions, we can still perform query operations in $O(\frac{\lg n}{\lg\lg n})$ time.

\begin{lemma}
The data structures in Section~\ref{sec:stringds} can support {\insertop} in $O(\frac{\lg n}{\lg\lg n})$ time such that for any constant $c_1>1$, after performing {\insertop} $c_1 n$ times, the operations {\accessop}, {\rankop} and {\selop} can still be supported in $O(\frac{\lg n}{\lg\lg n})$ time. 
\end{lemma}
\begin{proof}
As in the proof of Lemma~\ref{lem:cspsi}, we also modify the approach of Fleischer~\cite{f1996} here. 
It is however more challenging to apply Fleischer's approach here because we have to split the leaves in a specific way: A leaf has to be split into one leaf that contains one character only, and a second leaf that contains the rest. 
This allows us to update the CSPSI structure $E$ correctly. 
Note that this only applies to leaves, not internal nodes. 

To perform $\insertop_{\alpha}(S, i)$, after we locate the superblock, $j$, containing $S[i-1]$, recall that there are three cases. 
In the first case, superblock $j$ is not full, so we insert the character $\alpha$ into it. 
In the second case, superblock $j$ is full, but superblock $j-1$ is not. We remove the first character from superblock $j$ so that we can insert $\alpha$ into it, and then we insert the removed character into superblock $j-1$. 
In this case, superblock $j-1$ is the only superblock that has one more character after insertion. 
In the third case, superblock $j$ is full, but superblock $j-1$ is full or does not exist. We insert character $\alpha$ into superblock $j$, and then split it (in the proof of Lemma~\ref{lem:stringtheta}, we state that we create a new superblock before performing the insertion, which is equivalent to the process described here). 
We can unify the above three cases by letting $B$ be the leaf containing the superblock whose size increases before we consider whether we should split it. 
More specifically, $B$ contains superblock $j$ in the first and the third cases, and it contains superblock $j-1$ in the second. 
Thus, in our algorithm (to be presented), we need only describe the operations performed when we insert a character into $B$. 

To describe our algorithm for {\insertop}, we need some definitions. 
In our algorithm, we call a leaf {\em full} only when we mark it as full, and once we mark a leaf full, it is always considered as a full leaf. 
Note that in the three cases listed above, we need to check whether a certain superblock is full, and we say that a superblock is full when the leaf containing it is marked as full. 
If a leaf $B_1$ is full, but the leaf, $B_2$, immediately to its left is not, then $B_2$ is called an {\em overflow leaf} of $B_1$. 
Initially, the $i$\textsuperscript{th} leftmost leaf contains one character if $i$ is an odd number. 
If $i$ is an even number, we store $m = \lfloor 2L/\lg \sigma \rfloor$ characters in the $i$\textsuperscript{th} leftmost leaf, mark it as a full leaf, and we make the leaf immediately to the left to be its overflow leaf. 
For simplicity, we assume that $n$ is divisible by $m+1$. 
We define an internal node to be {\em big} when it has $2b$ or more children. 
 
The following is our algorithm (recall that we need only describe the operations performed when we insert a character into $B$, and we also omit the details of updating $E$ and the structures stored in internal nodes since the proof of Lemma~\ref{lem:stringtheta} already shows how to update them): 
\begin{enumerate}
\item Insert the new character into leaf $B$. Let $v$ be the node that $r_B$ points to, and if $B$ is an overflow leaf, let $B'$ be the leaf immediately to the right of $B$, i.e. $B$ is the overflow leaf of $B'$. 

\item If $v$ is big, then split $v$ into two smaller nodes $v'$ and $v''$, and mark all the child edges of $v'$ and $v''$ as unpaired (a child edge of a node is an edge between this node and one of its children). 

If the parent, $w$, of $v$ is also big, mark its edges to $v'$ and $v''$ as paired. 

If leaf $B$ is an overflow leaf, let $u$ be the node that $r_{B'}$ points to. If $u \ne v$, we perform the same operations on $u$ as above if it is big, which includes the operations on its (big) parent. 

\item If $v$ is the root of the tree, split $B$ into two leaves such that the left leaf, $B_l$, contains one character, while the right leaf, $B_r$ contains the rest. 
The $r$-pointers of both $B_l$ and $B_r$ point to their common father. 
Mark $B_r$ as a full leaf, and make $B_l$ be its overflow leaf. 
If $B$ is an overflow leaf of $B'$ before, now $B'$ no long has an overflow leaf. 

If $v$ is not the root of the tree, and we have inserted $\lceil\lg n\rceil/\lceil\lg\sigma\rceil$ characters into $B$ since the last time we update $r_B$, set $r_B$ to be the father of $v$, and if $B$ is an overflow leaf, also move $r_{B'}$ one level up the tree. 
\end{enumerate}

The way we split internal nodes in Step 2 is the same as that in \cite{f1996}, and the information about paired edges is used to decide how to split the internal node. 
The main idea of the above algorithm is to move the $r$-pointers of a full leaf, $B'$, and its overflow leaf, $B$, simultaneously, so that when the overflow leaf $B$ becomes full and has to be split, $r_{B'}$ already points to the root of the tree. The next insertion into $B'$ will split $B'$ since $r_{B'}$ points to the root. 

We now prove that, after $c_1 n$ insertions, the following conditions hold:
\begin{enumerate}
\item The B-tree is a $(b, 4b)$-tree. 
\item The height of the B-tree is $\Theta(\frac{\lg n}{\lg\lg n})$. The maximum height is $\frac{2\lg n}{\lg\lg n}$. 
\item Each leaf contains at most $2L$ bits. 
\item Each full leaf has $\Theta(L)$ bits. 
\item There do not exist two consecutive leaves that are not full. 
\end{enumerate}

To prove Condition~1, note that Fleischer~\cite{f1996} proved the correctness of his approach by showing a set of invariants holds after each step of his algorithm. 
Since we modify his algorithm by moving the $r$-pointers of a full leaf and its overflow leaf simultaneously, we essentially repeat some of the steps in his algorithm up twice. 
Thus a strict proof can be given by walking through his proof and making trivial changes. 
Condition~2 follows directly from Condition~1. 
Conditions~3 and 4 are true because a non-full leaf becomes full only when its $r$-pointer reaches the root, and each time we move up its $r$-pointer, roughly $\lg n$ bits have been inserted into the leaf. 
Finally, Condition 5 can be proved by induction. It is true initially before we perform the sequence of {\insertop} operations, and we always create a new leaf that is not full between two full leaves.

Conditions 1-3 guarantee that query operations can be supported in $O(\frac{\lg n}{\lg\lg n})$ time, while Conditions 4 and 5 guarantee that the space bound remains the same. Thus this lemma follows. 
\qed\end{proof}

We now consider the support for {\deleteop}. 
To support {\deleteop}, our approach here is to simply locate the leaf containing the character to be deleted and remove this character. 
After performing {\deleteop} $c_2 n$ times this way for any $0 < c_2 < 1$, the time required to answer a query will not change asymptotically. 
When we perform a mixed sequence of insertions and deletions, we perform each insertion as if there were no deletions performed in between. 
There is one technical detail here: When a leaf is made empty by deletions, we cannot remove it. 
This is because this leaf could be an overflow leaf, whose removal could affect future {\insertop} operations performed on the full leaf immediately to its right. 
Thus a sequence of deletions may leave a number of empty leaves. 
For each empty leaf $B_e$, we keep them conceptually in the data structure, but we deallocate the space used to store them, and store its $r$-pointer with its parent. This will not increase the space for internal nodes asymptotically. 
The corresponding child edge of its parent is set to be a NULL pointer, and when we count the number of superblocks stored in the subtree rooted at this parent, this empty leaf still counts, and its corresponding entry in $E$ is not removed. 
We store $r_{B_e}$ with the parent of $B_e$ because if $B_e$ is an overflow leaf, and in the future, one character is to be inserted into the full leaf to its right, we can still create a new superblock for $B_e$ so that this insertion can still be performed.


So far we have described our algorithms for {\insertop} and {\deleteop}, and showed that after $c_1 n$ insertions and/or $c_2 n$ deletions, a query operation can still be supported in $O(\frac{\lg n}{\lg\lg n})$ time. 
We can now describe a modified version of global rebuilding that will deamortize the support for {\insertop} and {\deleteop} completely, and different from the original global rebuilding approach of Overmars and van Leeuwen~\cite{ol1981}, our approach neither increases the space bounds asymptotically, nor requires any buffer. 
We thus call our approach {\em succinct global rebuilding}. 

As with the original global rebuilding approach, we start the rebuilding process when the number of {\insertop} and {\deleteop} operations performed exceeds half the initial length of the string stored in the data structure. 
Let $n_0$ denote the length of the string when we start rebuilding. 
In the new data structure, the value of $\lceil\lg n\rceil$ is determined by the value of $n_0$. 
We will build the new structure in the next $n_0/3$ {\insertop} and {\deleteop} operations. 

The main difference between our approach and the original approach in \cite{ol1981} is that during the process of rebuilding, we never store two copies of the same data, i.e. the string $S$. 
Instead, our new structure stores a prefix, $S_p$, of $S$, and the old data structure stores a suffix, $S_s$, of $S$. 
During the rebuilding process, each time we perform an insertion or deletion, we perform such an operation on either $S_p$ or $S_s$. 
After that, we remove the first $3$ characters from $S_s$, and append them to $S_p$. 
This would finish building the new data structure in $n_0 / 3$ update operations, if all the updates were performed on $S_p$. 
This is however not always the case. 
For the general case, we observe that the only operation that may make the rebuilding process take more time is insertions performed on $S_s$. 
Thus to speed up the process, after an insertion into $S_s$, we remove the first $4$ characters (instead of $3$) from $S_s$ and append them to $S_p$. 
Hence we can always finish rebuilding after we perform at most $n_0 / 3$ update operations. 
After rebuilding, the length of the string stored in the new data structure is at least $2n_0 /3$ and at most $4n_0/3$. 
This means our approach can handle the change of $n$ smoothly, as the difference between $\lceil\lg n_0\rceil$ and $\lceil\lg n\rceil$ is at most $1$.

The above process always keeps one copy of the string $S$, and thus all the leaves occupy $n\lg\sigma + o(n)\lg\sigma$ bits. 
The {\deleteop} operation described previously however does not remove internal nodes, although it deallocates the memory for empty leaves. 
To make it possible to deallocate the memory for internal nodes fast, we use the strategy in the proof of Lemma~\ref{lem:cspsi} that stores the internal nodes in a single memory block. 
A fraction of the space in this memory block is wasted, but it still uses $o(n)$ bits. 
Thus we can deallocate the memory block for the internal nodes of the old structure in one step after we finish building the new structure. 
The same strategy can be used to maintain the CSPSI structure $E$. We also build the universal table $A_2$ incrementally as in the proof of Lemma~\ref{lem:cspsi}.

The above analysis shows that our succinct global rebuilding scheme can deamortize the support for {\insertop} and {\deleteop} without changing the space bounds. 
During the rebuilding process, we use both old and new data structures to answer queries, and it is easy to verify that we can still support each query operation in $O(\frac{\lg n}{\lg\lg n})$ time.


\section{Reducing the Space Overhead when $w = \omega(\lg n)$}
\label{app:lgnchange}

To reduce the space overhead when $w = \omega(\lg n)$, we use two different strategies to maintain two different types of data. 
For the internal nodes in the B-tree, we note that the data structures constructed for them occupy $o(n)$ bits in total. 
Thus we can use the approach in the proof of Lemma~\ref{lem:cspsi} to store the internal nodes in a single memory block, and this only wastes $o(n)$ bits if each of the pointers that point to the external nodes can be encoded in  $O(\lg n)$ bits (we will show how to achieve this later). 
The same strategy also applies to the internal nodes of the B-tree for the CSPSI structure $E$. We can also use it for the blocks of the superblocks contained in the leaves of $E$; they occupy $o(n)$ bits so that we can afford wasting a constant fraction of the space. 
Since the size of each memory area allocated this way does not change for the same value of $\lceil\lg n\rceil$, we can make them adjacent in the memory to form a single memory area, and we call it the {\em fixed memory area} of the structure. 

We cannot do so for the superblocks of our string structure; they occupy $m=n\lg\sigma + o(n)\cdot\lg\sigma$ bits, so we cannot afford increasing this space cost by a constant factor. 
We use a different strategy: we divide the bits required to encode all the superblocks into $\sqrt{m/w}$ chunks of $\sqrt{mw}$ bits each. 
All the chunks are full except the last one, so only up to $\sqrt{mw}-1$ bits in these chunks may be wasted. 
We use $\sqrt{m/w}$ points of size $w$ to record the starting position of each chunks, and this costs $\sqrt{mw}$ bits. 
We divide each chunk into segments, and each segment is used to store a block of a superblock. 
Note that a segment may span up to two chunks so that all but the last chunk is full. 
To encode the starting position of a block, we only need to encode the rank of the chunk containing the first bit of the block which takes $\frac{1}{2}\lg(m/w)$ bits, and the offset within this chunk which takes $\frac{1}{2}\lg(mw)$ bits. 
Thus $O(\lg n)$ bits are sufficient to encode a pointer to a block or a superblock. 

The overall space cost of our structure is thus $n\lg\sigma + o(n)\cdot\lg\sigma + O(\sqrt{wm}) + O(w) = n\lg\sigma + o(n)\cdot\lg\sigma + O(\sqrt{wn\lg\sigma}) + O(w)$  bits. 
The third term in the above formula is $o(n)\cdot\lg\sigma$ if $w = o(n)\cdot\lg\sigma$, and it is $O(w)$ otherwise.\qed

\section{Proof of Theorem~\ref{thm:string}}
\label{app:string}

Ferragina~\etal~\cite{fmmn2007} showed how to use a generalized wavelet tree to extend results on strings over small alphabets to general alphabets. 
Let $q = \Theta(\sqrt{\lg n})$. 
This structure is essentially a $q$-ary tree of $O(\lceil\frac{\lg \sigma}{\lg\lg n}\rceil)$ levels, and each level stores a string over alphabet of size $q$. 
This can be directly applied to the dynamic case, and because each operation on $S$ requires a constant number of operations on each level, the time required for each operation is now $O(\frac{\lg n}{\lg\lg n}(\frac{\lg \sigma}{\lg\lg n}+1))$ for the entire structure. 
The only part that is not clear is the $O(w)$-bit term in space analysis: If we have an $O(w)$-bit term at each level of the wavelet tree, the space cost will be $O(w(\frac{\lg \sigma}{\lg\lg n}+1))$ bits. 
This cost can be decreased using the approach of M{\"a}kinen and Navarro~\cite{mn2008}: 
We observe that changes of $\lceil\lg n\rceil$ occur simultaneously for the sequences stored at different levels. 
Thus we can combine all the fixed memory areas into one area, and maintain the rest of the memory together ($m$ in Appendix~\ref{app:lgnchange} is now $nH_0 + o(n)\cdot\lg\sigma$). 
This reduces such space cost to $O(w)$. \qed

\end{document}